\documentclass[11pt]{amsart}
\pdfoutput=1
\usepackage{graphicx}
\usepackage{dsfont}
\usepackage{hyperref}
\usepackage{verbatim}
\usepackage[margin=1.2in]{geometry}
\newtheorem{theorem}{Theorem}[section]

\newtheorem{proposition}[theorem]{Proposition}

\theoremstyle{definition}

\newtheorem{remark}[theorem]{Remark}

\usepackage[toc,page]{appendix}

\numberwithin{equation}{section}

\newcommand{\IP}{\mathop{\null\mathds{P}}\nolimits}

\newcommand{\IR}{\mathds{R}}

\newcommand{\IZ}{\mathds{Z}}
\newcommand{\IC}{\mathds{C}}

\newcommand{\dist}{d}

\newcommand{\diam}{\text{\textnormal{diam}}}

\newcommand{\BLS}{\mathcal{L}}

\newcommand{\Sob}{\mathcal{H}}
\newcommand{\spin}{\sigma}

\begin{document}
\title[Spin systems from loop soups]{Spin systems from loop soups} 

\author{Tim van de Brug}
\author{Federico Camia}
\author{Marcin Lis}

\address{Tim van de Brug\\
Department of Epidemiology and Biostatistics\\
VU University Medical Center Amsterdam\\
De Boelelaan 1089a\\
1081 HV Amsterdam\\
The Netherlands}
\email{t.vandebrug@vumc.nl}

\address{Federico Camia\\
New York University Abu Dhabi\\
Saadiyat Island\\
Abu Dhabi\\
UAE\\
and
VU University Amsterdam\\
Department of Mathematics\\
De Boelelaan 1081a\\
1081 HV Amsterdam\\
The Netherlands}
\email{federico.camia@nyu.edu}

\address{Marcin Lis\\Statistical Laboratory\\ Centre for Mathematical Sciences\\University of Cambridge\\
Wilberforce Road\\Cambridge CB3 0WB\\UK}
\email{m.lis@statslab.cam.ac.uk}

\subjclass[2010]{82B20, 60G60, 60G18, 82B41}

\begin{abstract}
We study spin systems defined by the winding of a random walk loop soup. For a particular choice of loop soup intensity, we show that the
corresponding spin system is reflection-positive and is dual, in the Kramers-Wannier sense, to the spin system $\textnormal{sgn} (\varphi)$ where
$\varphi$ is a discrete Gaussian free field.

In general, we show that the spin correlation functions have conformally covariant scaling limits corresponding to the one-parameter family of
functions studied by Camia, Gandolfi and Kleban (\emph{Nuclear Physics B} {\bf 902}, 2016) and defined in terms of the winding of the Brownian
loop soup. These functions have properties consistent with the behavior of correlation functions of conformal primaries in a conformal field theory.
Here, we prove that they do correspond to correlation functions of continuum fields (random generalized functions) for values of the intensity
of the Brownian loop soup that are not too large.

\end{abstract}

\maketitle

\section{Introduction} \label{intro}
The random walk loop soup (RWLS) was introduced by Lawler and Trujillo Ferreras \cite{LawTru} as a discrete analog of the Brownian loop soup
(BLS) of Lawler and Werner \cite{LawWer}. The latter is a collection of planar loops of various sizes positioned at random, uniformly and independently,
within a planar domain. Each loop is a Brownian path constrained to begin and end at the same {root point}, but otherwise with no restriction,
and characterized by a \emph{time length} $t$ that is linearly related to its average area. The distribution in $t$ is proportional to $dt/t^{2}$,
so that there are many more small loops than large, and is chosen to ensure invariance under scale transformations. The overall density of loops
is characterized by a single parameter: the \emph{intensity} $\lambda>0$.

The BLS turns out to be not just scale invariant, but fully conformally invariant. For sufficiently low intensities $\lambda$, the intersecting loops form
clusters whose outer boundaries are distributed like Conformal Loop Ensembles (CLEs) \cite{She,SheWer}. CLEs are the unique ensembles of planar,
non-crossing and non-self-crossing loops satisfying a natural conformal restriction property that is conjecturally satisfied by the continuum scaling limits
of interfaces in two-dimensional models of statistical physics. The loops of a CLE$_\kappa$ are forms of SLE$_\kappa$ (the Schramm-Loewner
Evolution with parameter $\kappa$ \cite{Schramm2011,SheWer}). The CLEs generated by the BLS correspond to values of $\kappa$ between $8/3$
and~$4$. For example, it was recently proved that the collection of outermost interfaces in a planar critical Ising model in a finite domain
with plus boundary condition converges to CLE$_3$ in the scaling limit~\cite{BenHon}. Moreover,
the collection of outer boundaries of clusters of loops from the RWLS also converges to CLE in the scaling limit for appropriate values of the intensity $\lambda$~\cite{Lupu, BCL}.

Motivated by the
work of Freivogel and Kleban \cite{Freivogel} on bubble nucleation in theories of eternal inflation, Camia, Gandolfi and Kleban defined and computed certain statistical correlation functions that characterize aspects of the BLS distribution~\cite{CGK}.
They looked in particular at the net winding of all the loops around a given set of points
and found results consistent with the behavior of correlation functions of primary fields in a conformal field theory (CFT). 
The winding of Brownian paths and loops has been the subject of classical works of Spitzer~\cite{spitz}, Yor~\cite{yor}, and Pitman and Yor~\cite{pityor},
and more recently was studied e.g.\ by Garban and Trujillo Ferreras~\cite{gartru}.
It has also been discussed in the physics literature in connection to anyons (see e.g.~\cite{Ouv,CDO}).

Using the RWLS, in Section \ref{spin_fields} of this paper we introduce \emph{spin systems} which are discrete analogs of a one-dimensional subclass of the objects studied in \cite{CGK}. 
We note that similar objects are discussed in Section~6 of~\cite{Lej11}.
When the intensity of the RWLS is $1/2$, the corresponding spin system $\sigma$ is dual, in the Kramers-Wannier sense,
to the spin model $\textnormal{sgn} (\varphi)$ where $\varphi$ is a discrete Gaussian free field.
In Section \ref{griffiths} we show that $\sigma$  satisfies the Griffiths inequalities and is reflection-positive.
In Section \ref{sec:correlation_functions}, for general $\lambda$, we show that the correlation functions of the spin systems defined in Section \ref{spin_fields}
converge to the conformally covariant functions studied in \cite{CGK}. Our proof uses a convergence result of a certain observable related to the loop erased walk due Bene{\v{s}}, Lawler, and Viklund~\cite{BLV}.

In the last section, we show that, for values of the intensity $\lambda$ of the BLS that are not too large, but still including the most interesting case, $\lambda=1/2$,
one can construct continuum Euclidean fields (random generalized functions) whose correlation functions are the functions obtained in~\cite{CGK}. These fields
do not seem to correspond to any currently known CFT. As pointed out in~\cite{CGK}, the putative CFT associated with those correlation functions has the interesting
feature that the conformal dimensions of the primary operators are real and positive, but vary continuously and are periodic functions of a real parameter.

\section{RWLS spin fields and the discrete Gaussian free field} \label{spin_fields}

The \emph{(rooted) random walk loop measure} $\tilde \mu$ is a measure on nearest neighbor loops on $\IZ^2$ (possibly scaled by a factor $a>0$),
which we identify with loops in the complex plane by linear interpolation. For a loop $\gamma$ in $\IZ^2$, we define
\[
\tilde\mu(\gamma) = \frac{1}{t_{\gamma}} 4^{-t_{\gamma}},
\]
where $t_{\gamma}$ is the time length of $\gamma$, i.e.\ its number of steps. The (rooted) random walk loop soup $\mathcal{\tilde L}$ with intensity $\lambda>0$
(see \cite{LawTru}) is a Poissonian realization from the measure $\lambda\tilde\mu$. For $a>0$, by a \emph{discrete domain} in the scaled lattice $a\IZ^2$, we mean
a connected subgraph of $a\IZ^2$ which can be written as a union of square faces of $a\IZ^2$. For a discrete simply connected domain $D_a$ in $a\IZ^2$, let $\mathcal{\tilde L}_{D_a}$
be the collection of random walk loops in ${D_a}$, and let $\tilde \mu_{D_a}$ be the measure $\tilde \mu$ restricted to loops that stay in $D_a$.
By $\langle \cdot \rangle_{{D_a}}=\langle \cdot \rangle_{{D_a},\lambda}$ we will denote the expectation with respect to the loop soup with intensity measure
$\lambda \tilde \mu_{D_a}$.

Let $D^*_a$ denote the dual of $D_a$, whose vertices are the faces of $D_a$. For $z \in D^*_a$ we define
\[
\tilde N_{D_a}(z) = \sum_{\gamma\in \mathcal{\tilde L}_{D_a}} N_{\gamma}(z),
\]
where $N_{\gamma}(z)$ is the winding number of $\gamma$ around $z$. We define the random walk loop soup \emph{spin field} by
\[
\sigma(z)=\sigma_{D_a}(z) = e^{i \pi \tilde N_{D_a}(z)},
\]
which takes values $\pm 1$.

Let $\partial D_a\subset a\mathbb{Z}^2\setminus D_a$ be the set of vertices at graph distance $1$ from $ D_a$.
The discrete Gaussian free field (DGFF) on $D_a$ with boundary conditions $\psi$ is a multidimensional Gaussian variable $\varphi: D_a \cup \partial D_a \to \IR$ satisfying $\phi|_{\partial D_a} =\psi$ with density
given by
\[
\frac{1}{\mathcal{Z}} \exp \Big (-\frac12 \sum_{x \sim y} (\varphi_x-\varphi_y)^2 \Big)
\]
where the sum is over all edges $\{x,y\}$ in $D_a \cup \partial D_a$. Equivalently, $\varphi$ is a Gaussian variable with mean given by the harmonic extension of 
$\psi$, and covariance $\textnormal{Cov}(\varphi_x,\varphi_y)=G(x,y)$, where $G$ is the Green's function of simple random walk killed on hitting $\partial D_a$. 

The DGFF $\varphi$ can be thought of as a model of a random surface whose elevation above a point $x$ in the plane is given by $\varphi_x$, and where
large gradients between neighboring points are penalized. 
The continuum counterpart -- the Gaussian free field (GFF), to which the DGFF converges in an appropriate sense in the scaling limit, is too rough to be defined pointwise but can still be made sense of 
as a random generalized function. It is a Gaussian field with covariance given by the Green's function of two-dimensional Brownian motion, and, as such, is conformally invariant.
It turns out that the GFF is a universal object in random conformal geometry as its (carefully defined) level and flow lines encode different variants of the Schramm-Loewner Evolution curves~\cite{SchShe,MilShe1}.
It is also the scaling limit of other discrete models like the height function of the dimer model~\cite{Kenyon}.

A deep connection between random walk loop soups and the DGFF in form of the celebrated isomorphism theorems has its roots in the seminal work of Symanzik~\cite{Symanzik}.
We will focus on the result of Le Jan~\cite{Lej11} which may be viewed as an extension of Dynkin's isomorphism~\cite{Dyn84a,Dyn84b}. Consider a random walk loop soup $\tilde{\mathcal{L}}$.
Count the number of visits of all loops in $\tilde{\mathcal{L}}$ to a vertex $x$, and denote the number by $N_x$. The occupation field of $\tilde{\mathcal{L}}$ at $x$, denoted by $\mathcal{T}_x$, is a sum of $N_x$ (globally) independent 
exponential random variables with mean~$1$ together with $\frac12$ times one additional exponential variable with mean~$1$ (or in other words a $\text{Gamma}(1,\frac12)$ random variable). Le Jan~\cite{Lej11} proved that the joint
distribution of $(\mathcal{T}_x)_{x\in D_a}$ for a loop soup with intensity $\lambda=\frac12$ is equal to that of $\frac12(\varphi^2_x)_{x\in D_a}$.
Later Lupu~\cite{Lupu2016} provided a coupling between $\tilde{\mathcal{L}}$ and $\varphi$ that also accounts for the signs of~$\varphi$. His result is in turn intrinsically related to the Edwards--Sokal coupling 
between the Fortuin--Kasteleyn model with $q=2$ and the Ising model~\cite{ForKas,EdwSok,LupWer}. 

The Ising model~\cite{Ising} on $D_a$ is a random assignment $\mathbf{s}: D_a \to \{ -1,+1\}$ of spins to the vertices of $D_a$ drawn 
according to the law with the discrete density w.r.t.\ the counting measure given by
\[
\frac{1}{\mathcal{Z}'} \exp \Big (-\frac12 \sum_{x \sim y} J_{\{x,y\}} (\mathbf{s}_x-\mathbf{s}_y)^2 \Big),
\]
where the positive numbers $J_{\{x,y\}}$ are called coupling constants, the sum is taken over all edges $\{x,y\}$ of $D_a$, and ${\mathcal{Z}'}$ is the
normalizing constant. The relation with the DGFF is that the law of the sign $\textnormal{sgn}(\varphi)$ conditioned on the amplitude $|\varphi|$ is the
Ising model with free boundary conditions and coupling constants
\begin{align} \label{eq:dgffising}
J_{\{x,y\}}=|\varphi_x\varphi_y|.
\end{align}
A fundamental construction of Kramers and Wannier~\cite{KraWan} assigns to an Ising model with free boundary conditions on $D_a$ an Ising model with $+1$ boundary conditions on the dual domain $D_a^*$.
In the dual model, the spins are assigned to the vertices of $D_a^*$ (which are the faces of $D_a$), and the spin of the unbounded face is fixed to be $+1$.
For an edge $e$, let $e^*$ be the dual edge crossing $e$. The dual coupling constants satisfy the Kramers--Wannier duality relation:
\begin{align}\label{eq:kramerswannier}
J_{e^*} = -\frac12 \log \tanh J_{e}.
\end{align}

Recall that $\varphi$ is defined on the vertices, and $\spin$ on the faces of $D_a$.
In view of the following theorem, \eqref{eq:dgffising}, and \eqref{eq:kramerswannier}, the spin field $\spin$ with parameter $\lambda=\frac12$ can be thought of as a Kramers--Wanier dual of $\textnormal{sgn} (\varphi)$ -- the sign of a DGFF.
\begin{theorem}[Sampling $\sigma$ with $\lambda=\frac12$ from the DGFF and Ising model] \label{thm:DGFFIsing}
Consider the following algorithm:
\begin{itemize}
\item[(1)] Sample the amplitude of the DGFF $|\varphi|$ on $D_a$ with $0$ boundary conditions.
\item[(2)] Sample the Ising model on the dual domain $D_a^*$ with $+1$ boundary conditions, and coupling constants given by
\begin{align} \label{eq:dualcoupling}
J_{\{x,y\}^*} = -\frac12 \log \tanh |\varphi_x \varphi_y|.
\end{align}
\end{itemize}
The resulting assignment of $\pm 1$'s to the faces of $D_a$ is distributed like the spin field $\spin$ with parameter $\lambda=\frac12$.
\end{theorem}
\begin{proof}
Let $\tilde{\mathcal{L}}$ be a random walk loop soup in $D_a$ with intensity $\lambda=\frac12$ together with its occupation field $(\mathcal{T}_x)_{x\in D_a}$. For an edge $e$, let $N_e$ be the total number of unoriented traversals 
of $e$ by all loops in $\tilde{\mathcal{L}}$. It is known\footnote{~The first observation that this conditional distribution should be the one of a random current can be found in \cite{Wer2016} (the discussion after Proposition 7). 
However, the parameters of the current in~\cite{Wer2016} are incorrect. 
In Proposition 3.2 of~\cite{Lej16}, an almost correct statement is given (modulo a missing factor of $2$ resulting from the fact that each unoriented edge corresponds to two oriented edges) and the proof uses Gaussian integrals.
 } 
 that conditioned on the value of $(\mathcal{T}_x)_{x\in D_a}$, $(N_e)_{e\in E(D_a)}$ is distributed like a sourceless random current in $D_a$
with parameters (as defined in \cite{GHS,DC2016}, with $\beta=1$) 
\begin{align} \label{eq:coupling}
J_{\{ x,y\}} =  2\sqrt{\mathcal{T}_x \mathcal{T}_y}.
\end{align}

Note that the value of $\spin(z)$ is determined by the edges with odd values of $N_e$ in the following way: draw a simple path in the dual graph connecting $z$ 
with infinity and count the number of odd-valued edges that cross the path. Set $\spin(z)=-1$ if the resulting number is odd, and $\spin(z)=+1$ otherwise. 

It is a standard consequence of the Kramers--Wannier duality that the set of odd-valued edges in the random current is distributed
like the set of dual edges separating opposite spins in the Ising model dual to that defined by~\eqref{eq:coupling} (see e.g.\ \cite{DCL}). This means that conditioned on $(\mathcal{T}_x)_{x\in D_a}$, which by Le Jan's results 
is distributed like $\frac12(\varphi^2_x)_{x\in D_a}$, the spin model $\spin$ has the law of an Ising model on $D^*_a$ with coupling constants given by~\eqref{eq:dualcoupling}.
\end{proof}

\begin{proposition}[Sampling $\sigma$ with $\lambda=\frac12$ from the DGFF and coin flips] \label{thm:DGFFcoins}
Consider the following algorithm:
\begin{itemize}
\item[(1)] Sample the Gaussian free field $\varphi$ on $D_a$ with $0$ boundary conditions, and let $\eta$ be the set of dual bonds separating vertices with different values of $\textnormal{sgn}(\varphi)$.
\item[(2)] For each dual bond $\{x,y\}^*\notin \eta$, sample an independent Bernoulli random variable with success probability $\exp(-2 |\varphi_x \varphi_y|)$,
and let $\omega$ be the set of edges with a success outcome. 
\item[(3)] For each connected component of $\eta \cup \omega$ (treated as a subgraph of the dual graph $D_a^*$, including isolated vertices) that does not touch the outer boundary of $D_a^*$, sample an independent, symmetric $(\pm1)$-valued random variable,
and assign its value to each face of $D_a$ in that connected component. Assign $+1$ to the remaining faces.
\end{itemize}
The resulting configuration of $\pm 1$'s on the faces of $D_a$ is distributed like the spin field $\spin$ with parameter $\lambda=\frac12$.
\end{proposition}
\begin{proof}
From Theorem~\ref{thm:DGFFIsing} and the Edwards--Sokal coupling between the (FK) random cluster model with parameter $q=2$ and the Ising model (see e.g.\ ~\cite{Grimmett}), it is enough 
to prove that the set $\eta$ after step (2) is distributed like the random cluster model on $D^*_a$ with wired boundary conditions and parameters 
\[
p^*_{\{x,y\}^*}=1-\exp(-2J_{\{x,y\}^*})= 1- \tanh |\varphi_x \varphi_y|.
\]

To this end, consider an FK model on $D_a$ with free boundary conditions and parameters $p_e$.
Recall that in the Edwards--Sokal coupling, to recover the random cluster model configuration from the Ising spin configuration $\textnormal{sgn}(\varphi)$ 
with coupling constants $J_e$ satisfying $p_e=1-\exp(-2J_e)$, 
one performs independent Bernoulli percolation with success probabilities $p_e$ on the edges
whose endpoints carry the same spin. The dual random cluster configuration with wired boundary conditions 
is hence distributed like independent percolation on $D^*_a$ with success probabilities $1-p_e=\exp(-2J_e)$ 
union with the dual edges separating opposite spins.
Hence, $\eta$ after step (2) is distributed like a random cluster model on $D_a^*$ with wired boundary conditions and success probabilities (see, for example,
equation~(6.5) of~\cite{Grimmett})
\[
p^*_{\{x,y\}^*} = \frac{2-2p_{\{x,y\}}}{2-p_{\{x,y\}}}=\frac{2\exp(-2|\varphi_x \varphi_y|)}{1+\exp(-2|\varphi_x \varphi_y|)}=1- \tanh |\varphi_x \varphi_y|
\]
which completes the proof.
\end{proof}

\begin{remark}
To a configuration of loops one can naturally assign a $1$-form, i.e., an antisymmetric function on the directed edges of $D_a$ given by
the difference of the total number of jumps of the loops along the directed edge and its reversal. It is clear that this $1$-form is divergence free in the sense that 
the sum of values over all directed edges emanating from a single vertex is zero. This makes it posssible to define a \emph{height function} of the collection of loops
by summing up the total flux of the $1$-form across paths in the dual graph as it is done e.g.\ to define the height function of a dimer model.
In this language, the field $\tilde N_{D_a}(z)$ is exactly the height-function, and our construction of the spin system is analogous to the relation 
of the XOR-Ising model and the height function of a related dimer model~\cite{Dub, BdT, DCL}.
\end{remark}
\section{Griffiths inequalities and reflection positivity} \label{griffiths}
In this section we prove that the spin field is positively correlated for all $\lambda >0$, and its ``massive'' version is reflection positive at $\lambda=\frac{1}2$.

The following inequalities satisfied by the spin field are classical in the context of ferromagnetic spin systems~\cite{Gri67,KelShe68}.

\begin{proposition}[Griffiths inequalities] \label{prop:griffiths} 
Let $D_a\subset D_a'$ be any two finite discrete domains in $a\IZ^2$.
Let $\lambda>0$, and $z_1,\ldots, z_n$, $w_1, \ldots, w_k$, be faces of $D_a$. Then,
\begin{itemize}
\item[(i)] $\langle \prod_{j=1}^n\sigma(z_j) \rangle_{D_a} \geq 0$,
\item[(ii)] $\langle \prod_{j=1}^n\sigma(z_j) \rangle_{D_a} \geq \langle \prod_{j=1}^n\sigma(z_j) \rangle_{D'_a} $,
\item[(iii)] $\langle \prod_{j=1}^n\sigma(z_j)  \prod_{j=1}^k\sigma(w_j)  \rangle_{D_a}  \geq \langle \prod_{j=1}^n\sigma(z_j) \rangle_{D_a}  \langle \prod_{j=1}^k\sigma(w_j) \rangle_{D_a} $.
\end{itemize}
\end{proposition}
\begin{proof}
Using the definition of the spin field and the expression of the $n$-point function in the first displayed equation in the proof of Theorem~4.3 of~\cite{CGK},
we can express the $n$-point function as
\begin{align} \label{eq:n-point-function}
 \langle\sigma (z_1) \ldots \sigma(z_n) \rangle_{D_a} &=    \big \langle (-1)^{\sum_{j=1}^n \tilde N_{D_a}(z_j)} \big \rangle_{D_a} \nonumber \\
 & = \exp\Big [\lambda  \tilde \mu \big((-1)^{\sum_{j=1}^n N_{\gamma}(z_j)}-1\big) \Big] \nonumber \\
 &=  \exp\Big[ -2\lambda\tilde \mu_{D_a}\big(\gamma: \textstyle \sum_{j=1}^n N_{\gamma}(z_j) \text{ is odd}\big) \Big] ,
\end{align}
from which the first two inequalities immediately follow. Moreover, we get that
\begin{align*}
&\frac{\langle \prod_{j=1}^n\sigma(z_j)  \prod_{j=1}^k\sigma(w_j)  \rangle_{D_a} }{\langle \prod_{j=1}^n\sigma(z_j) \rangle_{D_a}  \langle \prod_{j=1}^k\sigma(w_j) \rangle_{D_a} } =\\ 
&\qquad  \exp\Big[ 4 \lambda\tilde \mu_{D_a}\big(\gamma: \textstyle \sum_{j=1}^n N_{\gamma}(z_j) \text{ is odd, and }\textstyle \sum_{j=1}^k N_{\gamma}(w_j)\text{ is odd} \big) \Big]\ge 1,
\end{align*}
which gives the third inequality.
\end{proof}

Note that the spin field cannot be directly defined on the whole square lattice $a\IZ^2$ due to the fact that large random walk loops carry infinite mass and
hence each face of $a\IZ^2$ is covered by infinitely many loops. With the help of the second inequality from the theorem above, we can define an infinite volume
limit field as the finite domain $D_a $ approaches $\IZ^2$. However, by analyzing the correlation functions and using e.g.\ the fact that the mass of random walk
loops in $a\IZ^2$ passing through a single edge is infinite,  we see that the field is trivial, i.e., it is a collection of iid symmetric $(\pm 1)$-valued variables.

One can get around this issue by considering a {massive} version of the loop measures~\cite{Cam17}.
Let $\kappa >0$. For a loop $\gamma$ in $\IZ^2$, we define the \emph{massive loop measure} by
\[
\tilde\mu^{\kappa}(\gamma) = \frac{1}{t_{\gamma}} (4+\kappa)^{-t_{\gamma}}.
\]
Under $\tilde\mu^{\kappa}$, the total measure of large loops  intersecting a bounded region of space decays exponentially with the size of the loops.
One can hence define a spin field $\sigma^{\kappa}$ directly from the infinite volume loop soup with intensity measure $\lambda \tilde\mu^{\kappa}(\gamma)$.

We will now show that this spin model is reflection positive. (See \cite{Biskup} for more information on the concept and use of reflection
positivity in the context of lattice spin models. The question of reflection positivity in the loop soup context is addressed in Chapter 9 of \cite{Lej11}).
$\IZ^2$ has a natural reflection symmetry along any line $l$ going through a set of dual vertices. Such a line splits $\IZ^2$ in two halves, $\IZ^2_{+}$
and $\IZ^2_{-}$. We also split accordingly the dual graph $(\IZ^2)^*$ in two halves, $(\IZ^2_{+})^*$ and $(\IZ^2_{-})^*$, such that
$(\IZ^2_{+})^* \cap (\IZ^2_{-})^* = V^*_l$, where $V^*_l$ is the set of vertices of $(\IZ^2)^*$ that lie on $l$.

Let ${\mathcal F}^+$ (respectively, ${\mathcal F}^-$) denote the set of all functions of the spin variables
$(\sigma(z))_{z \in (\IZ^2_{+})^*}$ (respectively, $(\sigma(z))_{z \in (\IZ^2_{-})^*}$).
Let $\vartheta$ be the reflection with respect to $l$. With a slight abuse of notation, it induces a map $\vartheta : {\mathcal F}^{\pm} \to {\mathcal F}^{\mp}$ given by $\vartheta f(\sigma)= f(\sigma \circ \vartheta)$, $f\in {\mathcal F}^{\pm} $. Let $\langle \cdot \rangle_{\IZ^2}^{\kappa}$
denote expectation with respect to the infinite volume loop soup with intensity measure $\lambda \tilde\mu^{\kappa}$.
\begin{proposition} \label{prop:reflection_positivity}
For all $\kappa>0$ and $\lambda=1/2$, the infinite volume massive spin field $\sigma^{\kappa}$ is reflection-positive, i.e.,
for all functions $f,g \in {\mathcal F}^+$, $\langle f \vartheta g \rangle_{\IZ^2}^{\kappa} = \langle g \vartheta f \rangle_{\IZ^2}^{\kappa}$ and $\langle f \vartheta f \rangle_{\IZ^2}^{\kappa} \geq 0$.
\end{proposition}
\begin{proof}
The proof uses the Markov property of simple random walk loops described in~\cite{Wer2016}, and is analogous to that of Lemma~8.1 of \cite{CamLis} where a spin field for the non-backtracking
loop soup was defined. An alternative way to prove the result is to notice that the spin field is a function of the edge-occupation field of the loop soup that was shown to have a Markov property by
Le Jan~\cite{Lej16}.
\end{proof}

\section{Convergence of correlation functions} \label{sec:correlation_functions}

In this section we prove that the correlation functions of the spin fields converge, in the scaling limit, to conformally covariant functions. These are the functions
studied in \cite{CGK} and, as we explain in the next section, they are the correlation functions of the corresponding continuum winding fields.  We also study
the effect of a small perturbation of the boundary of the domain on the value of the field.


To state our results, we need to introduce the Brownian loop measure and the Brownian loop soup which are continuum analogs of the notions from the previous sections. 
A \emph{(rooted) loop} $\gamma$ of time length $t_{\gamma}$ is a continuous function $\gamma : [0,t_{\gamma}] \to \IC$ with $\gamma(0)=\gamma(t_\gamma)$.
Given a domain $D \subset \IC$, a conformal map $f : D \to \IC$, and a loop $\gamma$ in~$D$,
we define $f \circ \gamma$ to be the loop $f (\gamma)$ with time parametrization given by the Brownian scaling $f \circ \gamma (s) = f(\gamma(t))$, where
\[
 \quad s= s(t) = \int_0^{t}  | f'(\gamma(u))|^2 du,
\]
and $t_{f\circ \gamma} = s(t_{\gamma})$.
In particular, if $\Phi_{a,b}(w)=aw+b$, $a \neq 0$, then $\Phi_{a,b} \circ \gamma$ is the loop $\gamma$ scaled by $|a|$, rotated around the origin by $\arg a$ and
shifted by~$b$, with time parametrization $s(t)= |a|^2t$, and time length $t_{\Phi_{a,b}\circ \gamma} = |a|^2 t_{\gamma} $.

By $\mu^{br}$ we denote the \emph{complex Brownian bridge} measure, i.e., a probability measure on loops rooted at~$0$ of time length $1$ induced by the process $B_t=W_t-tW_1$, ${t\in[0,1]}$, 
where $W_t$ is a standard complex Brownian motion starting at $0$. For $z\in \IC$ and $t>0$, by $\mu^{br}_{z,t}$ we denote the complex Brownian bridge measure on loops rooted at $z$ of time length $t$, i.e.,
the measure
\[
\mu^{br}_{z,t} = \mu^{br} \circ \Phi_{\sqrt{t},z}^{-1}.
\]
The \emph{Brownian loop measure} is a $\sigma$-finite measure on loops given by
\[
\mu= \int_{\IC} \int_0^{\infty} \frac{1}{2\pi t^2} \mu^{br}_{z,t} dt dA(z).
\]
This measure is clearly translation invariant and it is easy to check that it is scale invariant. This means that $\mu = \mu \circ \Phi_{a,b}$ for any $a>0$ and $b\in\IC$.
Since $\mu$ inherits rotation invariance from the complex Brownian motion, we actually have that $\mu = \mu \circ \Phi_{a,b}$ for any $a,b\in\IC$, $a\neq 0$.
To recover the full \emph{conformal invariance} of Brownian motion one has to consider $\mu$ as a measure on \emph{unrooted loops}, i.e., equivalence classes of loops 
under the relation $\gamma \sim  \theta_r \gamma $ for every $r \in \IR$, where $\theta_r \gamma(s) = \gamma(s+r \ \text{mod} \ t_{\gamma})$.

If $D$ is a domain, then by $\mu_D$ we denote the measure $\mu$ restricted to loops which stay in $D$. Let $D,D'$ be two simply connected domains, and let $f: D\to D'$ 
be a conformal equivalence. The full conformal invariance of $\mu$ is expressed by the fact that $\mu_{D'} \circ f =\mu_D$.
A proof of this can be found in~\cite{Law}.

The \emph{Brownian loop soup} $\BLS_D=\BLS_{D,\lambda}$ with intensity parameter $\lambda >0$ is a Poissonian collection of loops with intensity measure $\lambda  \mu_D$.
We write $\BLS=\BLS_{\IC}$. The Brownian loop soup inherits all invariance properties of the Brownian loop measure. In particular, $\BLS_{D'}$ has 
the same distribution as $f[\BLS_{D}]$.

Let $r(D,z)$ denote the conformal radius of $D$ seen from $z$. For a mesh size $a>0$, let $D_a$ be the largest discrete domain in $a\IZ^2$ that is contained in~$D$,
and $D^*_a$ its dual. For $z \in D$, let $z^a \in D^*_a$ be a dual vertex closest to $z$ (chosen in any deterministic way if there is more than one such vertex).
Let $\Delta = \lambda/8$.
\begin{theorem}[Convergence of the $n$-point function]\label{thm:npoint} Let $D$ be a simply connected bounded Jordan domain.
The limit 
\[
\lim_{a\to 0} a^{-2n\Delta} \langle \sigma(z^a_1) \ldots \sigma(z^a_n) \rangle_{D_a} =: \psi_D(z_1,\ldots,z_n)^{2\lambda}
\]
exists, and there is a positive constant $c<\infty$ such that
\begin{align*}
& \psi_D(z_1,\ldots,z_n) = c^n \prod_{j=1}^n r(D,z_j)^{-1/8} \exp[ \mu_D(\gamma : N_{\gamma}(z_j) \text{ is odd}, | \overline\gamma \cap \{z_1,\ldots,z_n\} | \geq 2 ) ] \\
& \qquad \qquad \qquad \qquad \quad \exp\Big[ - \mu_D\big(\gamma: \textstyle \sum_{j=1}^n N_{\gamma}(z_i) \text{ is odd},| \overline\gamma \cap \{ z_1,\ldots,z_n\} |\geq 2\big) \Big]. \\
\end{align*}
Moreover, if $f:D\to D'$ is a conformal map, then
\[
\psi_{D'}(f(z_1),\ldots,f(z_n)) = \psi_D(z_1,\ldots,z_n) \prod_{j=1}^n |f'(z_j)|^{-1/8}.
\]
\end{theorem}
\begin{remark}
If $n=2$, then $\psi_D$ can be expressed as
\begin{align*}
&\psi_D(z_1,z_2) = c^2 r(D,z_1)^{-1/8} r(D,z_2)^{-1/8} \exp[ 2 \mu_D(\gamma: N_{\gamma}(z_1) \text{ and } N_{\gamma}(z_2) \text{ are odd})].
\end{align*}
\end{remark}

\begin{proposition}[Boundary perturbations]\label{boundaryperturbations}
Let $D'\subset \mathds{D}$ be a simply connected subset of the unit disk containing 0. For a mesh size $a>0$, let $\sigma_{\mathds{D}_a}(z^a)$ be the spin field generated by $\mathcal{\tilde L}_{\mathds{D}_a}$, and let $\sigma_{D'_a}(z^a)$ be the spin field generated by the loops in $\mathcal{\tilde L}_{\mathds{D}_a}$ that stay in $D'_a$. In this coupling,
\[
\lim_{a\to 0} \IP( \sigma_{\mathds{D}_a}(0^a) = \sigma_{D'_a}(0^a)) =: \chi(r(D',0))
\]
exists, and
\[
\frac{d \chi(r)}{dr} \bigg|_{r=1} = \frac{\lambda}{8},
\]
which equals the scaling dimension $\Delta$ of the winding field.
\end{proposition}

\begin{proof}[Proof of Theorem \ref{thm:npoint}]
We will write the $n$-point function as the product of 1-point functions and a factor that only depends on loops that are macroscopic in the scaling limit. We will use a result of Bene{\v{s}}, Lawler and Viklund~\cite{BLV} 
to determine the asymptotics of the 1-point functions. We will then use the coupling between the Brownian loop measure and the random walk loop measure of Lawler and Trujillo Ferreras \cite{LawTru} to compute the limit of the remaining factor.

Using the expression \eqref{eq:n-point-function} from the proof of Proposition \ref{prop:griffiths}, the 1-point function can be written as
\[
\langle \sigma(z^a) \rangle_{D_a} = \exp[-2\lambda \tilde \mu_{D_a} ( \gamma : N_{\gamma}(z^a) \text{ is odd}) ].
\]
Hence,
\begin{align*}
& \frac{ \langle \sigma(z^a_1) \ldots \sigma(z^a_n) \rangle_{D_a} } { \prod_{j=1}^n \langle \sigma(z^a_j) \rangle_{D_a} } \\
 & =  \exp\Big[ - 2\tilde \mu_{D_a}\big(\gamma: \textstyle \sum_{j=1}^n N_{\gamma}(z^a_j) \text{ is odd},| \overline\gamma \cap \{ z^a_1,\ldots,z^a_n\} |\geq 2\big) \Big]  \\
& \qquad \qquad \qquad \prod_{j=1}^n \frac{\exp[ -2\lambda \tilde \mu_{D_a}(\gamma: N_{\gamma}(z^a_j) \text{ is odd}, \overline\gamma \cap \{z^a_1\ldots,z^a_n\} = \{z^a_j\}) ]}  {\exp[-2\lambda \tilde\mu_{D_a}(\gamma : N_{\gamma}(z^a_j) \text{ is odd})]} \\
& =  \exp\Big[ - 2\tilde \mu_{D_a}\big(\gamma: \textstyle \sum_{j=1}^n N_{\gamma}(z^a_j) \text{ is odd},| \overline\gamma \cap \{ z^a_1,\ldots,z^a_n\} |\geq 2\big) \Big]  \\
&\qquad \qquad \qquad \prod_{j=1}^n \exp[ 2\lambda\tilde\mu_{D_a}(\gamma : N_{\gamma}(z^a_j) \text{ is odd}, | \overline\gamma \cap \{z^a_1,\ldots,z^a_n\} | \geq 2 ) ] \\
& \to   \exp\Big[ - 2 \mu_D\big(\gamma: \textstyle \sum_{j=1}^n N_{\gamma}(z_j) \text{ is odd},| \overline\gamma \cap \{ z_1,\ldots,z_n\} |\geq 2\big) \Big] \\ 
&\qquad \qquad \qquad  \prod_{j=1}^n \exp[ 2\lambda\mu_D(\gamma : N_{\gamma}(z_j) \text{ is odd}, | \overline\gamma \cap \{z_1,\ldots,z_n\} | \geq 2 ) ],
\end{align*}
where the convergence holds in the scaling limit $a\to 0$. To justify the convergence, note that the sets of loops that appear in the last expression only contain loops that cover at least two points of $\{z_1,\ldots,z_n\}$. With probability one, in the Brownian loop soup in $D$ there are only finitely many loops covering at least two points of $\{z_1,\ldots,z_n\}$. The distance between these loops and the points is positive with probability one. Hence, if the Brownian loops are approximated sufficiently well by random walk loops, then the winding numbers of the Brownian loops around $z_1,\ldots,z_n$ will be the same as the corresponding winding numbers of the approximating random walk loops. The convergence now follows from the strong coupling of \cite{LawTru} between the Brownian loop soup measure and the random walk loop soup measure.

To prove convergence of the $n$-point function, it remains to show convergence of the 1-point function. 
Recall that $r(D,z)$ denotes the conformal radius of $D$ seen from $z$. By Theorem 1.4 of \cite{BLV} there exist $u>0$ and $0<c_1<\infty$ such that
\begin{align*}
 a^{-2\Delta} \langle \sigma(z^a_j) \rangle_{D_a} & = a^{-\lambda/4} \exp[-2 \lambda\tilde \mu_{D_a}(\gamma: N_{\gamma}(z^a_j) \text{ is odd})] \\
& = a^{-\lambda/4} \big[ c_1 ( a^{-1} r(D_a,z^a_j) )^{1/4} [ 1 + O (  ( a^{-1} r(D_a,z^a_j) )^{-u} ) ] \big]^{-\lambda} .
\end{align*}
As $a \to 0$, $ r(D_a,z^a_j) \to r(D,z_j)$, which implies that
\[
\lim_{a \to 0} a^{-2\Delta} \langle \sigma(z^a_j) \rangle_{D_a} = \tilde c_1 r(D,z_j)^{-\lambda/4}.
\]
This completes the proof of convergence of the $n$-point function, i.e.\ the first statement of the theorem.

Finally, the conformal covariance of $\psi_D(z_1,\ldots,z_n)$ easily follows from the conformal invariance of the Brownian loop measure and of the conformal radius. 
\end{proof}

\begin{proof}[Proof of Proposition \ref{boundaryperturbations}]
Let the spin fields $\sigma_{\mathds{D}_a}$ and $\sigma_{D'_a}$ be coupled as stated in the proposition. We have that
\begin{align*}
\IP( \sigma_{\mathds{D}_a}(0^a) &= \sigma_{D'_a}(0^a)) \\
& = \IP(\#(\gamma\in \mathcal{\tilde L}_{\mathds{D}_a} : N_{\gamma}(0^a) \text{ is odd}, \gamma \setminus D' \not= \emptyset) = 0) \\
&\quad + \IP(\#(\gamma\in \mathcal{\tilde L}_{\mathds{D}_a} : N_{\gamma}(0^a) \text{ is odd}, \gamma  \setminus D' \not= \emptyset) = 2k \text{ for some } k\geq 1) \\
&= \exp[ - \lambda \tilde \mu_{\mathds{D}_a} ( \gamma: N_{\gamma}(0^a) \text{ is odd}, \gamma  \setminus D' \not= \emptyset ) ] \\
&\quad + O( [\lambda \tilde \mu_{\mathds{D}_a} ( \gamma: N_{\gamma}(0^a) \text{ is odd}, \gamma  \setminus D' \not= \emptyset)]^2 ),
\end{align*}
as $r(D',0)\to 1$. This follows from the fact that the random variable $\#(\gamma\in \mathcal{\tilde L}_{\mathds{D}_a} : N_{\gamma}(0^a) \text{ is odd}, \gamma  \setminus D' \not= \emptyset)$ has a Poisson distribution with mean $\lambda \tilde \mu_{\mathds{D}_a} ( \gamma: N_{\gamma}(0^a) \text{ is odd}, \gamma \setminus D' \not= \emptyset)$.

Let $f:\mathds{D} \to D'$ denote the conformal map from $\mathds{D}$ onto $D'$ such that $f(0)=0$ and $f'(0)>0$. By the convergence of the random walk loop soup to the Brownian loop soup \cite{LawTru}, and using Proposition 3 of \cite{Wer08} and Lemma A.2 from \cite{CGK}, we have
\begin{align*}
 \lim_{a\to 0} \tilde \mu_{\mathds{D}_a} (\gamma : N_{\gamma}(0^a) &\text{ is odd}, \gamma \setminus D' \not= \emptyset) \\
& = \mu_{\mathds{D}} (\gamma : N_{\gamma}(0) \text{ is odd}, \gamma \setminus D' \not= \emptyset) \\
&= \sum_{k=-\infty}^{\infty} \mu_{\mathds{D}} (\gamma : N_{\gamma}(0) =2k+1, \gamma \setminus D' \not= \emptyset) \\
& = \sum_{k=-\infty}^{\infty} \frac{1}{2\pi^2 (2k+1)^2} \log \frac{1}{f'(0)} \\
 &= - \frac{1}{8} \log r(D',0).
\end{align*}
Hence,
\begin{align*}
\lim_{a\to 0} \IP & (\tilde V_{\mathds{D}_a}(0^a) = \tilde V_{D'_a}(0^a))   =r(D',0)^{\lambda/8} + O\Big(\Big[\frac{\lambda}{8} \log r(D',0)\Big]^2\Big) =: \chi(r(D',0))
\end{align*}
exists. The last statement of the proposition follows immediately.
\end{proof}

\section{Brownian loop soup winding fields}
In this section we show that, for values of the intensity $\lambda$ of the BLS that are not too large, but still including e.g.\ the case $\lambda=1/2$ for $\pm 1$-valued fields,
one can construct continuum Euclidean fields (random generalized functions) whose correlation functions are the functions obtained in~\cite{CGK}.

Following~\cite{CGK}, we define a winding field arising from the Brownian loop soup (see also Section 6 of~\cite{Lej11}). We will restrict our attention to bounded domains $D$.
Let $\bar{\gamma}$ be the \emph{hull} of the loop $\gamma$, i.e., the complement of the unique unbounded connected component of the complement of $\gamma$.
Here, as we will often do, we treat $\gamma$ as a subset of $\IC$. We say that $\gamma$ \emph{covers}  $z$ if $z\in  \overline{ \gamma}$.
We will be interested in quantities defined in terms of the \emph{total winding} of all loops of the loop soup around any given point $z$. Since $\BLS$
is scale invariant, $\{ \gamma \in \BLS: \gamma \ \text{covers} \ z, \ \diam \gamma \leq \delta \}$ is infinite almost surely for all $\delta >0$. This forces us to regularize the loop soup so 
that only finitely many loops cover each point. One way to do this is to introduce the \emph{``ultraviolet''} cutoff $\delta$ on the size of the loops by defining
\[
\BLS^{\delta}_D =\{ \gamma \in \BLS_D:  \diam \gamma > \delta \}.
\]
Similarly, by $\mu^{\delta}_D$ we denote the measure $\mu_D$ restricted to loops of diameter larger than $\delta$.
Note that each point $z$ is covered by only finitely many loops from $\BLS^{\delta}_D$ almost surely since the Brownian loop measure of such loops is finite.
We can now define
\begin{align} \label{eq:winddef}
N(z)=N^{\delta}_D(z) = \sum_{\gamma \in \BLS^{\delta}_D} N_{\gamma}(z),  
\end{align}
where $ N_{\gamma}(z)$ is the {winding number} of $\gamma$ around $z$ (if $z\in \gamma$, then we put $N_z(\gamma)=0$). Note that the loops which do not cover $z$ do not contribute to the above sum, and therefore the sum is finite almost surely.
The {winding field} is then defined by
\[
V^{\delta}(z) =V^{\delta}_{\beta}(z) =e^{i \beta N(z)}, \qquad \beta \in [0,\pi],
\]
The correlation functions of this random field were explicitly computed in~\cite{CGK} in the limit as $\delta \to 0$.
In particular, it was proved that the one-point function $\langle V^{\delta}(z) \rangle$ decays like 
$\delta^{2\Delta}$ where
\[
\Delta = \lambda \frac{\beta(2\pi-\beta)}{8\pi^2}.
\]

Note that since $|\delta^{-2\Delta} V_{\beta}^{\delta}(z) | = \delta^{-2\Delta}$ for all $z\in D$, the field $ \delta^{-2\Delta}V_{\beta}^{\delta}$ does not converge as a function on $D$ as $\delta \to 0$. 
Hence, to obtain convergence results, one has to treat $ \delta^{-2\Delta}V_{\beta}^{\delta}$ as an element of a topological space larger than any classical function space. This is usually achieved by thinking of $ \delta^{-2\Delta}V_{\beta}^{\delta}$ as a random distribution, 
i.e., a random continuous functional on some appropriately chosen space of test functions where the action of $ \delta^{-2\Delta}V_{\beta}^{\delta}$ on a test function 
$f$ is given by
\[
  \delta^{-2\Delta}V_{\beta}^{\delta} ( f ) =\delta^{-2\Delta} \int_D  V_{\beta}^{\delta}(z) {f (z)} dz.
\]

A convenient framework describing such functionals is given by Sobolev spaces with negative index, which we here briefly recall, following~\cite{Dub09}.
Let $\mathcal{H}_0^1=\mathcal{H}_0^1(D)$ be the classical Sobolev Hilbert space, i.e., the closure of $C^{\infty}_0(D)$ in the norm
\[
\| f \|^2_{\mathcal{H}_0^1}:= \int_D | \nabla f (z) |^2 dz.
\]
 Let $u_1,u_2,\ldots$  be the eigenfunctions and $\lambda_1 < \lambda_2 \leq \ldots\to\infty$ the respective eigenvalues of the positive definite Laplacian on $D$ with Dirichlet boundary conditions.
We assume that $u_1,u_2,\ldots$ are normalized to have unit norm in $L^2=L^2(D)$. These eigenfunctions form orthogonal bases for both $\mathcal{H}_0^1$ and $L^2$, and if 
$f =\sum_{i=1}^{\infty} a_i u_i \in \mathcal{H}_0^1\subset L^2$, then
\[
\| f \|^2_{\mathcal{H}_0^1}= \sum_{i=1}^{\infty}\lambda_i a_i^2.
\]
One can by analogy define for any $\alpha>0$ the Hilbert space $\mathcal{H}_0^{\alpha}$ as the closure of $C^{\infty}_0(D)$ with respect to the norm
\[
\| f \|^2_{\mathcal{H}_0^{\alpha}}= \sum_{i=1}^{\infty}\lambda_i^{\alpha} a_i^2.
\]
The Sobolev space $\mathcal{H}^{-\alpha}$ is then defined as the Hilbert dual of $\mathcal{H}_0^{\alpha}$, i.e., the space of continuous linear functionals $h$ on $\mathcal{H}_0^{\alpha}$
with norm 
 \[
 \| h \|_{\mathcal{H}^{-\alpha}}  = \sup_{f: \ \| f \|_{\mathcal{H}^{\alpha}_0 }=1} | h(f) |.
 \]
Note that $L^2 \subset \mathcal{H}^{-\alpha}$, where the action of $h = \sum_{i} a_i u_i \in L^2$ on $f \in \mathcal{H}^{\alpha} $ is given by $h(f)= \int_D  h(z)  {f (z)} dz$.
It is easily checked that in this case
\begin{align} \label{eq:norm}
\| h \|^2_{\mathcal{H}^{-\alpha}} = \sum_{i} \frac1{\lambda_i^{\alpha}}{a_i^2}.
\end{align}

In our last main result we address the question asked in~\cite{CGK} about the existence of winding fields as random generalized functions:
\begin{theorem} \label{thm:fieldconv}
Let $D$ be a bounded, simply connected domain with a smooth boundary, and let $\Delta < 1/2$. Then for every $\alpha > 3/2$,
the field $V_{\beta}^{\delta}$ treated as a random distribution converges  as $\delta \to 0$ in
second mean in the Sobolev space $\Sob^{-\alpha}$, i.e.,
there exists a random distribution $V_{\beta} \in \Sob^{-\alpha}$ measurable with respect to $\BLS_D$, such that 
\[
\big\langle \| \delta^{-2\Delta}V_{\beta}^{\delta} -V_{\beta} \|^2_{\Sob^{-\alpha}} \big\rangle_D \to 0  \qquad \text{as } \delta \to 0.
\]
\end{theorem}

We also show the {non-triviality} of winding fields, i.e., that they are not Gaussian.
\begin{proposition} \label{prop:non-triviality}
The conformally covariant functions derived in \cite{CGK} as limits of the winding field $n$-point functions do not satisfy Wick's relations for Gaussian fields. 
\end{proposition}
The proof of this result is postponed until the end of this section.

The idea for the proof of Theorem~\ref{thm:fieldconv} is to show that the fields $V^{\delta}_{\beta}$ form a Cauchy sequence in the Banach space
$L^2(\Omega,\mu_D; \Sob^{-\alpha})$ of $\Sob^{-\alpha}$-valued, $\mu_D$-square integrable random variables.
The main ingredient is the following proposition describing the behavior of the two-point functions. We give the proof after the proof of Theorem~\ref{thm:fieldconv}.

\begin{proposition} \label{prop:2pointfunc}
Let $D$ be a bounded simply connected domain with a $C^1$ boundary.
For any $z,w \in D$, $z\neq w$, and $\Delta >0$, the limit
\begin{align} \label{eq:convtwopoint}
 \lim_{\delta,\delta' \to 0^+} (\delta\delta')^{-2\Delta} \langle V^{\delta}_{\beta}(z) \overline{V^{\delta'}_{\beta}(w)}\rangle_{ D}=:  \langle V_{\beta}(z) \overline{V_{\beta}(w)}\rangle_{ D} 
\end{align}
exists. Moreover, if $\Delta <1/2$, then the convergence holds also in $L^1(D\times D, dzdw)$, and if $\Delta \geq 1/2$, then $\langle V_{\beta}(z) \overline{V_{\beta}(w)}\rangle_{ D}\notin L^1(D\times D, dzdw)$.
\end{proposition}

We are now ready to prove the convergence of the field in the appropriate Sobolev space.
\begin{proof}[Proof of Theorem~\ref{thm:fieldconv}]
Define the rescaled field $\tilde V^{\delta}_{\beta}= \delta^{-2\Delta}V^{\delta}_{\beta}$, and note that $\tilde V^{\delta} \in L^2$ since it is bounded. By \eqref{eq:norm} and orthonormality of $u_i$'s in $L^2$,
we have
\begin{align*}
& \big\langle \| \tilde V^{\delta}_{\beta}-\tilde V^{\delta'}_{\beta}\|^2_{\Sob^{-\alpha}} \big\rangle_D \\&
= \sum_{i} \frac{1}{\lambda_i^{\alpha}} \big\langle \big | \int_D  (\tilde V^{\delta}_{\beta}- \tilde V^{\delta'}_{\beta}) \overline{u_{i}(z)} dz \big|^2  \big\rangle_D \\
& = \sum_{i}\frac{1}{ \lambda_i^{\alpha}} \Big\langle  \int_D\int_D(\tilde V^{\delta}_{\beta}(z)-\tilde V^{\delta'}_{\beta}(z))  \overline{u_{i}(z) (\tilde V^{\delta}_{\beta}(w)-\tilde V^{\delta'}_{\beta}(w))}  {u_{i}(w)} dz dw\Big\rangle_D \\
&\leq \Big( \int_D\int_D \big| \langle (\tilde V^{\delta}_{\beta}(z)-\tilde V^{\delta'}_{\beta}(z))\overline{(\tilde V^{\delta}_{\beta}(w)-\tilde V^{\delta'}_{\beta}(w))} \rangle_D \big| dzdw\Big)     \sum_{i}\frac{c^2}{ \lambda_i^{\alpha-1/2}} ,
\end{align*}
where in the last inequality we used Fubini's theorem and the uniform bound~$\|u_i\|_{L^{\infty}(D)} \leq c \lambda_i^{1/4}$ of~\cite{Grieser}.
The series in the last expression is finite for $\alpha > 3/2$ by Weyl's law~\cite{Weyl} which says that
\[
N(\ell) = \frac{1}{4\pi} \text{Area} (D) \ell (1+o(1)) \qquad \text{as } \ell \to \infty,
\]
where $N(\ell) = \# \{ i: \lambda_i \leq \ell\}$ is the eigenvalue counting function.
Therefore the convergence in $L^1(D\times D, dzdw)$ of the two-point functions from Proposition~\ref{prop:2pointfunc} implies that 
\begin{align} \label{eq:Cauchy}
 \big\langle \|\tilde V^{\delta}_{\beta}-\tilde V^{\delta'}_{\beta}\|^2_{\Sob^{-\alpha}} \big\rangle_D \to 0  \qquad \text{as } \delta, \delta' \to 0^+.
\end{align}

The space $L^2(\Omega,\mu_D; \Sob^{-\alpha})$ of random variables $X$ with values in $ \Sob^{-\alpha}$ and such that $\langle \|X \|^2_{\Sob^{-\alpha}} \rangle_D <\infty$ 
is a Banach space with norm $\langle \|X \|^2_{\Sob^{-\alpha}}\rangle^{1/2}$ \cite{LedTal}. Therefore, the desired convergence follows from \eqref{eq:Cauchy} and from the
completeness of Banach spaces. 
\end{proof}

\begin{proof}[Proof of Proposition~\ref{prop:2pointfunc}]
Assume that $\delta > \delta'$ and let $z,w \in D$. For $k,l \in \IZ$, let
\begin{align*}
\kappa^{\delta}_{z,w} &= \sum_{k,l \in \IZ} \mu^{\delta}_{D}\big ( \gamma : N_{\gamma} (z) = k, N_{\gamma} (w) = l,\ z,w \in \bar{ \gamma} \big)(1-\cos((k-l)\beta)),\\
\tau^{\delta}_{z,w} &= \sum_{k \in \IZ} \mu^{\delta}_{D} \big( \gamma : N_{\gamma} (z) = k,\ z \in \bar \gamma, w \notin \bar \gamma \big) (1-\cos(k\beta)),\\
\tau^{\delta',\delta}_{w,z} &= \sum_{k \in \IZ} \mu^{\delta',\delta}_{D} \big( \gamma : N_{\gamma} (w) = k,\ z,w \in \bar \gamma \big) (1-\cos(k\beta)),
\end{align*}
where $\mu^{\delta',\delta}_{D}$ is the measure $\mu_{D}$ restricted to loops of diameter $>\delta'$ and $\leq\delta$.

Using elementary properties of Poisson point processes, just as in the proof of Theorem~4.3 in~\cite{CGK}, we can express the two-point function in terms of the 
loop measure. We get that, for all $z, w$ and $\delta,\delta'$,
\begin{align} \label{eq:2pointformula}
&(\delta\delta')^{-2\Delta}  \langle V^{\delta}_{\beta}(z) \overline{V^{\delta'}_{\beta}(w)}\rangle_{ D} = (\delta\delta')^{-2\Delta}\exp
 \big( -\lambda (\kappa_{z,w}^{\delta} + \tau^{\delta}_{z,{w}}+\tau^{\delta'}_{w,{z}}+\tau^{\delta',\delta}_{w,z}) \big).
\end{align}
Let $m_{z,w} = \dist (z, \partial D \cup \{ w\})$.
By Lemma~3.2 of~\cite{CGK}, for $\delta < m_{z,w}$, 
\begin{align}  \label{eq:1pointformula}
\exp \Big( -\lambda  \tau^{\delta}_{z,{w}}\Big) = \Big(\frac{m_{z,w}}{\delta}\Big)^{-2\Delta} 
\exp ( -\lambda \tau^{m_{z,w}}_{z,{w}} ).
\end{align}
An analogous identity holds when we interchange $z$ and $w$. 
Moreover, note that $\kappa^{\delta}_{z,w} =\kappa^{|z-w|}_{z,w}  $ for $\delta < |z-w|$, and $\tau^{\delta',\delta}_{w,z} =0$ for $\delta <  |z-w|$.
Hence, \eqref{eq:2pointformula} and \eqref{eq:1pointformula} imply that if $\delta < m_{z,w}$ and $\delta' < m_{w,z}$, then
\begin{align} \label{eq:limitonepoint}
 (\delta\delta')^{-2\Delta}  \langle V^{\delta}_{\beta}(z) \overline{V^{\delta'}_{\beta}(w)}\rangle_{ D} = ( m_{z,w} m_{w,z})^{-2\Delta} 
 \exp ( -\lambda ( \kappa^{|z-w|}_{z,w}  +\tau^{m_{z,w}}_{z,{w}}+\tau^{m_{w,z}}_{w,{z}}) ) 
\end{align}
which is independent of $\delta$ and $\delta'$. This proves~\eqref{eq:convtwopoint}. 

We now assume that $\Delta <1/2$ and focus on the convergence in $L^1(D\times D, dzdw)$.
Observe that, by~\eqref{eq:2pointformula} and~\eqref{eq:1pointformula}, for all $z, w$ and $\delta,\delta'$,
\begin{align*}
(\delta\delta')^{-2\Delta}  \langle V^{\delta}_{\beta}(z) \overline{V^{\delta'}_{\beta}(w)}\rangle_{ D} &
 \leq (\delta\delta')^{-2\Delta} \exp \big ( -\lambda (\tau^{\delta}_{z,{w}}+\tau^{\delta'}_{w,{z}})\big) \\ 
&\leq (\delta \lor m_{z,w})^{-2\Delta} (\delta' \lor m_{w,z})^{-2\Delta} \\
& \leq (m_{z,w} m_{w,z})^{-2\Delta}. 
\end{align*}
Hence, by~\eqref{eq:convtwopoint} and dominated convergence, it is enough to show that
\begin{align*} 
 \int_{D} \int_D(m_{z,w}m_{w,z})^{-2\Delta}dzdw  < \infty.
\end{align*} 
Let $B_r(z) = \{ w: |z-w| < r \}$, $\mathbb{D}=B_1(0)$, and $m^{\mathbb{D}}_{z,w}=\dist(z,\partial \mathbb{D} \cup \{ w\})=|z-w|  \land (1-|z|)$. Let $f: \mathbb{D} \to D$ be a conformal equivalence. 
By the Koebe quarter theorem, $m^{\mathbb{D}}_{z,w} | f'(z) | \leq 4 m_{f(z),f(w)}$.
By Theorem 3.5 of~\cite{Pommerenke}, $f'$ has a continuous 
extension to $\overline{\mathbb{D}}$, and in particular, $\| f'\|_{L^{\infty}(\mathbb{D})} < \infty$.
By integration by substitution, the above integral is equal to
\begin{align*}
&\nonumber \int_{\mathbb{D}} \int_{\mathbb{D}}(m_{f(z),f(w)}m_{f(w),f(z)})^{-2\Delta} |f'(z)f'(w)|^2dzdw  \leq  4^{4\Delta} \| f'\|_{L^{\infty}(\mathbb{D})}^{4-4\Delta} \int_{\mathbb{D}} \int_{\mathbb{D}}(m^{\mathbb{D}}_{z,w}m^{\mathbb{D}}_{w,z})^{-2\Delta} dzdw. \\
\end{align*}
Since $(m^{\mathbb{D}}_{z,w})^{-2\Delta} \leq |z-w|^{-2\Delta} + (1-|z|)^{-2\Delta} $, it is now enough to note that
\begin{align*}
&\int_{\mathbb{D}} |z-w|^{-4\Delta} dz \leq\int_{B(w;2)} |z-w|^{-4\Delta} dz  = \int_0^{2\pi} \int_0^{2}  r^{1-4\Delta} dr d\theta < \infty , \\
 &\int_{\mathbb{D}}  (1-|z|)^{-2\Delta} dz = \int_{0}^{2\pi} \int_0^1 r(1-r)^{-2\Delta}dr d\theta < \infty, \\
 &\int_{\mathbb{D}} \int_{\mathbb{D}}  |z-w|^{-2\Delta }(1-|z|)^{-2\Delta}dwdz \leq \int_{2\mathbb{D}} |w|^{-2\Delta} dw  \int_{\mathbb{D}}  (1-|z|)^{-2\Delta} dz < \infty. 
 \end{align*} 
 This proves convergence in $L^p(D\times D, dzdw)$ for $\Delta < 1/(2p)$.
 
We now assume that $\Delta \geq 1/2$. 
 Note that, by scale invariance of the Brownian loop soup, and since the collection of outer boundaries of loops in the Brownian loop soup is thin in the sense of~\cite{NacWer}, 
 \begin{align*}
 \sup_{z,w \in D} \tau^{m_{z,w}}_{z,w}& \leq 2 \sup_{z,w \in D} \mu_{\IC}^{m_{z,w}}(\gamma: z \in \bar \gamma, w \notin \bar \gamma, \bar  \gamma \cap \partial D = \emptyset) \\
&  \leq 2 \sup_{z,w \in \IC} \mu_{\IC}^{|z-w|}(\gamma: z \in \bar \gamma, w \notin \bar \gamma) \\
& = 2 \mu^1_{\IC}  (\gamma: 0 \in \bar \gamma, 1 \notin \bar \gamma) < \infty .
 \end{align*}
Let $\rho =\diam D/2$. If $| z-w | > \rho$, then $\kappa_{z,w}^{\delta} \leq 2 \mu_{\IC}^{\rho}( \gamma \subset D) < \infty$. Hence, by~\eqref{eq:limitonepoint},  
 there exists $C>0$ such that
\begin{align*}
& (\delta \delta')^{-2\Delta} \langle V^{\delta}_{\beta}(z) \overline{V^{\delta'}_{\beta}(w)}\rangle_{ D} \geq C (m_{z,w} m_{w,z})^{-2\Delta}
\end{align*} 
if $ (z,w) \in I_{\delta,\delta'} := \{ (z,w) \in D^2:\delta < m_{z,w}, \delta' < m_{w,z}, |z-w|>\rho \}$. Using the fact that, for $z,w \in \mathbb{D}$, $m^{\mathbb{D}}_{z,w} \| f' \|_{\infty} \geq m_{f(z),f(w)}$, 
and again integrating by substitution, we have a lower bound of the form
\begin{align*}
 C \| f' \|_{{L^{\infty}(\mathbb{D})}}^{4-4\Delta}& \iint_{f^{-1}[I_{\delta,\delta'}]} (m^{\mathbb{D}}_{z,w} m^{\mathbb{D}}_{w,z})^{-2\Delta} d z d w \\
 & \geq C\| f' \|_{{L^{\infty}(\mathbb{D})}}^{4-4\Delta}  \iint_{f^{-1}[I_{\delta,\delta'}]}  ((1-|z|)(1-|w|))^{-2\Delta}  d z d w  \to \infty,
\end{align*}
as $ \delta,\delta'\to 0$ since $f^{-1}[I_{\delta,\delta'}] \nearrow f^{-1}\{(z,w) \in D^2 : |z-w| > \rho \}$. 
This shows that $\langle V_{\beta}(z) \overline{V_{\beta}(w)}\rangle_{ D}\notin L^1(D\times D, dzdw)$ for $\Delta \geq 1/2$.
\end{proof}

\begin{proof}[Proof of Proposition~\ref{prop:non-triviality}]
It is enough to provide one example. To that end, consider three points, $x,y,z$, contained in the unit disc $\mathbb{D}$ and at equal distance from the center of the disc and from each other.
We write $\phi(x,y,z)$ for the three-point function of the winding field in $\mathbb{D}$, and use analogous notation for the two- and one-point functions.

If the functions defined above satisfied the Wick's relations for Gaussian fields, they would in particular satisfy the following identity:
\begin{equation} \label{eq:gaussian-identity}
\frac{\phi(x,y,z)}{\phi(x)\phi(y)\phi(z)} - \frac{\phi(x,y)}{\phi(x)\phi(y)} - \frac{\phi(y,z)}{\phi(y)\phi(z)} - \frac{\phi(x,z)}{\phi(x)\phi(z)} + 2 =0 .
\end{equation}

Let $a$ denote the $\mu_{\mathbb{D}}$-measure of all loops that wind an odd number of times around $x,y$ and $z$.
Let $b$ denote the $\mu_{\mathbb{D}}$-measure of all loops that wind an odd number of times around two of $x,y$ and $z$, and an even number
of times around the remaining point.
An easy calculation using Theorem~\ref{thm:npoint} shows that \eqref{eq:gaussian-identity} can be written as
\[
e^{4\lambda(a+b)} \big[ e^{8\lambda b} - 3 \big] +2 =0 .
\]
However, the last equation cannot be always satisfied because one can, for example, make $b$ as large as desired by taking $x,y$ and $z$ close to $0$.
This concludes the proof.
\end{proof}

\bigbreak\noindent\textbf{Acknowledgments} 
TvdB and ML thank New York University Abu Dhabi for the hospitality during a visit in 2017.
TvdB and FC thank the Indian Statistical Institute, Delhi Station for the hospitality during a visit in 2016.
FC and ML thank Yves Le Jan for several interesting discussions. The authors also thank an anonymous referee for a careful reading of the manuscript and for providing useful comments.
Part of the research was conducted while TvdB was at the Department of Mathematics of Vrije Universiteit Amsterdam.
The research of ML was funded by EPSRC grants EP/I03372X/1 and EP/L018896/1.

\bibliographystyle{amsplain}
\bibliography{rlf}

\end{document}